\pdfoutput=1
\documentclass[a4paper,aps,pra,twocolumn,nofootinbib,10pt,superscriptaddress]{revtex4-1}

\usepackage[utf8]{inputenc}
\usepackage[british]{babel}
\usepackage[T1]{fontenc}
\usepackage[sc,osf]{mathpazo}
\usepackage[scaled=0.86]{berasans}
\usepackage[scaled=1.03]{inconsolata}
\usepackage[colorlinks,linkcolor=magenta,urlcolor=blue]{hyperref}
\usepackage{graphicx}
\usepackage[babel]{microtype}
\usepackage{amsmath,amssymb,amsthm}
\usepackage{mathtools}
\usepackage{xspace}
\usepackage{pgfplots}

\DeclareMathOperator{\tr}{tr}

\DeclareMathOperator{\erf}{erf}

\let\originalleft\left
\let\originalright\right
\renewcommand{\left}{\mathopen{}\mathclose\bgroup\originalleft}
\renewcommand{\right}{\aftergroup\egroup\originalright}

\newcommand{\bra}[1]{\left\langle #1 \right|}
\newcommand{\ket}[1]{\left| #1 \right\rangle}
\newcommand{\braket}[2]{\left\langle #1 \middle| #2 \right\rangle}
\newcommand{\ketbra}[2]{\left|#1\middle\rangle\middle\langle#2\right|}
\newcommand{\exval}[3]{\left\langle #1 \middle| #2 \middle| #3 \right\rangle}

\newcommand{\norm}[1]{\left\|#1\right\|}

\newcommand{\abs}[1]{\left|#1\right|}
\newcommand{\mean}[1]{\left\langle#1\right\rangle}
\newcommand{\comm}[2]{\left[#1,#2\right]}

\newcommand{\de}[1]{\left(#1\right)}

\newcommand{\DE}[1]{\left\{#1\right\}}

\newcommand{\chsh}{\mathcal{B}}

\newcommand{\Hi}{\mathcal{H}}

\newcommand{\id}{\mathbb{1}}

\newcommand{\R}{\mathbb{R}}
\DeclareMathOperator{\dint}{d\!}

\newcommand{\ie}{\emph{i.e.}\@\xspace}
\newcommand{\etal}{\emph{et al.}\@\xspace}

\newtheorem{theorem}{Theorem}
\newtheorem{lemma}[theorem]{Lemma}

\mathtoolsset{centercolon}

\hypersetup{pdftitle={Maximal violations and efficiency requirements for Bell tests with photodetection and homodyne measurements}}
\begin{document}
\title{Maximal violations and efficiency requirements for Bell tests with photodetection and homodyne measurements}

\author{Marco Túlio Quintino} 
\author{Mateus Araújo} 
 \affiliation{Departamento de Física, Universidade Federal de Minas Gerais,
 Caixa Postal 702, 30123-970, Belo Horizonte, MG, Brazil}
\author{Daniel Cavalcanti} 
\affiliation{Centre for Quantum Technologies, National University of Singapore, 3 Science drive 2, Singapore 117543}
\author{Marcelo França Santos}
 \affiliation{Departamento de Física, Universidade Federal de Minas Gerais,
 Caixa Postal 702, 30123-970, Belo Horizonte, MG, Brazil}
\author{Marcelo Terra Cunha} 
 \affiliation{Departamento de Matemática, Universidade Federal de Minas Gerais,
 Caixa Postal 702, 30123-970, Belo Horizonte, MG, Brazil}
\date{\today}

\begin{abstract}
	We study nonlocality tests in which each party performs photodetection and homodyne measurements. The results of such measurements are dichotomized and a Clauser-Horne-Shimony-Holt (CHSH) inequality is used. We prove that in this scenario the maximal violation is attainable and fully characterize the set of maximally violating states. If we restrict our search to states composed of at most $2$, $4$, and $6$ photons per mode, we find critical photodetection efficiencies of $0.48$, $0.36$, and $0.29$. We also found an entangled variation of the famous cat states that has critical efficiency $0.32$. These values are well within the limit of current photodetector technology, which suggests the present approach as a road for a loophole-free Bell experiment.
\end{abstract}

\maketitle

\section{Introduction}

	Since Bell proved his theorem in 1964 \cite{bell64}, there has been considerable interest in experimentally ruling out local hidden variables models.  Although Aspect's 1982 experiment \cite{aspect82} gave a strong evidence in favor of the existence of nonlocal correlations,
	it relied on the \emph{fair sampling} assumption, thus opening up the possibility of a local hidden variables description \cite{pearle70} for his experiment.

 	From a fundamental point of view there is no reason to believe that nature maliciously disrespects fair sampling. However, the recent advent of device-independent protocols \cite{ekert91,acin07,pironio10,colbeck11,bardyn09,rabelo11} gave an additional motivation to search for a loophole free Bell test.
In this case one may be fighting against an active opponent who can use the undetected photons to crack a given protocol. Hence, closing the detection loophole is a requirement for a demonstration of device-independent quantum information processing.

	In the standard Bell test using discrete variables and the Clauser-Horne-Shimony-Holt (CHSH) inequality \cite{ch74}, an overall detection efficiency higher than $2/3$ is required to close the detection loophole \cite{eberhard93,larsson01}. More recently, it was shown that the use of
higher dimensional entangled states (and other inequalities) can lower this requirement~\cite{Tamas-Stefano-Nicolas}. However, these experimental situations are still very demanding.
	
	An alternative method proposed to close the detection loophole in photonic systems is the use of homodyne measurements \cite{reid98}, which can be made very efficient. However, earlier results relying only on homodyne measurements required unfeasible states \cite{munro99,wenger03,ecava07,acin09,he10} or displayed very small violations \cite{nhacar04,grangier04}, indicating that homodyning alone may not render the definite Bell test. 
	
	More recently, Cavalcanti \etal\ explored a hybrid Bell test scenario that combines photodetection and homodyne measurements. An experimentally reasonable violation of a CHSH inequality was found in a setup involving a feasible state, although with detection efficiencies still comparable to the best numbers found in the discrete variable cases \cite{cavalcanti10}. 
	
	The main goal of this paper is to show hybrid schemes that overcome Cavalcanti \etal's result in two senses: larger violations and lower required efficiencies. First, we prove that the maximal violation of the CHSH inequality can indeed be found within the hybrid scenario. Moreover we fully characterize the set of states that attain this maximal value. Second, we study the robustness of the CHSH violations of natural classes of quantum states for several sources of errors (photodetection inefficiency, transmission losses, and dark counts). We demonstrate the existence of states that achieve both a large violation of the CHSH inequality and a high resistance to detection inefficiencies.  
	
	We organize our paper as follows:
	\begin{enumerate}
	\item Sec.\ \ref{sec:methods} introduces the standard Bell test scenario and the CHSH inequality.
	\item Sec.\ \ref{sec:hybrid} describes the hybrid measurements scenario, which involves homodyne measurement and photodetection in each side of the Bell test.
	\item In Sec.\ \ref{sec:maximal} we show that maximal violation of the CHSH inequality can be obtained in the present scenario and characterize the family of states achieving such violation. 
	\item As the family of states achieving maximal CHSH violation is physically unreasonable, we provide, in Sec.\ \ref{sec:restriction}, a study of other families of quantum states previously discussed in the literature, such as N$00$N states and truncated-Fock states. We show that some of these states can provide quite high CHSH violation. 
	\item Sec.\ \ref{sec:loophole} studies typical errors involved in the Bell test, such as detection inefficiencies, transmission losses, and dark counts. In special we find some quantum states which are very robust against photodetection inefficiency.
	\item Sec.\ \ref{sec:multipartite} briefly discusses the multipartite case.
	\item Finally Sec.\ \ref{sec:discussion} is devoted to some discussions and future directions. 
	\end{enumerate}
	


\section{The CHSH scenario}\label{sec:methods}

	Consider two parties, Alice and Bob, who can perform two possible measurements of two outcomes each. Alice's measurements will be labelled by $A_i$ ($i=0,1$) and can return possible results $a_i=\pm1$. Similarly, Bob can choose measurements $B_j$ ($j=0,1$) with possible outcomes $b_j=\pm 1$. The CHSH inequality imposes a constraint on the correlations attainable by any local hidden-variable theory, and can be expressed as
	\[\abs{E_{00} +E_{01} + E_{10} -E_{11}}\leq2,\]
	where the correlations $E_{ij}=p(a_i=b_j|A_i, B_j)-p(a_i\neq b_j|A_i, B_j)$, being $p(a_i=b_j|A_i, B_j)$ the probability that the outcomes of Alice and Bob are equal if measurements $A_i$ and $B_j$ are chosen.

	In quantum mechanics we can write the correlation terms as $E_{ij}=\tr(\rho A_i\otimes B_j)$, where $A_i$ and $B_j$ are quantum observables with eigenvalues $\pm 1$ and $\rho$ is the quantum state of the bipartite system.  
	Thus, the CHSH inequality can be written, within quantum mechanics, as
	\[
	\abs{\mean{\chsh}}\leq 2,
	\]
	where \[\chsh:=A_0\otimes B_0 + A_0\otimes  B_1 + A_1\otimes B_0 - A_1\otimes  B_1\] is the CHSH operator.
		
	The advantage of defining the operator in this way is that to find the state $\ket{\psi}$ that maximally violates the CHSH inequality one only has to find the norm of the CHSH operator and its corresponding eigenvector, so an unstructured search in the state space is unnecessary. For a more generous introduction we suggest \cite{wernerwolf}.

\section{CHSH with photodetection and homodyne measurements}\label{sec:hybrid}

	As explained before, the CHSH scenario involves two measurements of two outcomes per party. Here we are interested in the case where the observables chosen by Alice and Bob are given by the $X$ quadrature $X=\int_{-\infty}^{\infty} x\ketbra{x}{x}\dint x$ and the number of photons $N =\sum_{n=0}^{\infty}n\ketbra{n}{n}$, where $\ket{n}$ is a Fock state. 

	Both observables have an infinite number of possible outcomes, so we need to do a binning process in order to use them in a CHSH test, that is, map their outcomes into $+1$ and $-1$. The dichotomic version of the $N$ operator is the detection operator $D$, defined as
	\begin{subequations}
	\begin{equation} D := P_{D+} -P_{D-},\end{equation}
	where
	\begin{equation}
	P_{D+} := \sum_{n=1}^{\infty} \ketbra{n}{n}\quad\text{and}\quad P_{D-} := \ketbra{0}{0},
	\end{equation}
	\end{subequations}
	for which a click outputs the value $+1$ and the absence of a click outputs the value $-1$. This definition has a very clear physical motivation, since $D$ models photodetectors used for low intensity fields.
	
	For the $X$ operator we will define a dichotomic operator $Q$ that will output $+1$ if the $X$ measurement returns a value of $x$ inside a set $A^+$, and $-1$ if it returns a value in the complement $A^- = \mathbb{R}\setminus A^+$. So we define the operator $Q$ as 
	\begin{subequations}\label{eq:qdef}
	\begin{equation} Q := P_{Q+} - P_{Q-},\end{equation} 
	where 
	\begin{equation}
	P_{Q\pm}:=\int_{A^\pm} \ketbra{x}{x}\dint x.
	\end{equation}
	Note that $P_{D+} + P_{D-} = \id = P_{Q+} + P_{Q-}$.

	We can now calculate the associated matrix elements in the Fock basis:	
	\begin{align}
	\exval{m}{Q}{n} &= \exval{m}{P_{Q+}}{n}-\exval{m}{P_{Q-}}{n} \\
			&= 2\exval{m}{P_{Q+}}{n} -\delta_{mn} \\
			&= 2\int_{A^+}\varphi_m^*\varphi_n - \delta_{mn},
	\end{align}
	\end{subequations}
	where $\varphi_n(x)=\braket{x}{n}$ is the $n$th Hermite function, that is, the $n$th eigenstate of the $N$ operator in the position representation.
	
	As a matter of fact, note that the measurement operator $X$ restricted to the $\DE{\ket{0},\ket{1}}$ subspace after the sign binning process (\ie, $A^+=\R^+$), is given simply by
	\[ X = \sqrt{\frac{2}{\pi}}\sigma_x,\] 
	where $\sigma_{x}$ is the standard Pauli matrix. Also note that the measurement operator of an arbitrary quadrature $\cos(\theta)X+\sin(\theta)P$ in the same subspace and considering the same binning, where $P$ is the quadrature orthogonal to $X$, is given by
	\[ \sqrt{\frac{2}{\pi}}\de{\vphantom{\sqrt{\frac{2}{\pi}}}\cos(\theta)\sigma_x+\sin(\theta)\sigma_y}.\] In other words, if one applies the sign binning and deals with states in the subspace $\DE{\ket{0},\ket{1}}$, measuring a quadrature is equivalent to performing a spin measurement in the $XY$ plane. This fact will be useful, for instance, to study the violation of Bell inequalities in the multipartite scenario (see section \ref{sec:multipartite}).

	Using $Q$ and $D$, we can now define an operator
	\begin{equation}\label{eq:chsh}
	\chsh(A^+) := Q\otimes Q + Q \otimes D + D \otimes Q- D \otimes D
	\end{equation}
	so that the problem of finding the maximum violation of the corresponding inequality reduces to finding\footnote{Where $\norm{A}:=\sup{_{\norm{\ket{\psi}}=1}\norm{A\ket{\psi}}}$. In finite dimensions $\norm{A}$ is just the largest singular value of $A$.}
	\[
	\sup_{A^+} \norm{\chsh(A^+)},
	\]
	in other words, the choice of binning that maximizes the norm of the CHSH operator.

	To actually solve this maximization problem we need to search through generic subsets of $\mathbb{R}$, which is a difficult task. Therefore we will choose $A^+$ to be an interval, with arbitrary endpoints\footnote{Unstructured search indicates that this is usually the optimal case, or very near to it.}. In this case we can even evaluate $\chsh$ explicitly, simplifying the numerical work involved. In principle we could have considered different sets and hence different binnings for Alice and Bob, but in all our calculations we found no advantage in doing so.

\section{Maximal violations}\label{sec:maximal}

	In this section we establish a direct connection between the maximal reachable violation and the binning choice (set $A^+$). In appendix \ref{sec:russo} we prove that states that attain maximal violations always belongs to the subspace generated by $\DE{\ket{0},\ket{\Xi}}^{\otimes 2}$, where 
\[ \ket{\Xi} := \frac{1}{\sin\theta}\sum_{n=1}^\infty 2\int_{A^+}\varphi_0\varphi_n\ket{n} \]
and $\theta \in (0,\pi)$ is a function of the binning, defined via
\[ \cos\theta := 2\int_{A^+}\varphi_0^2-1.\]

 	Restricted to $\operatorname{span}\DE{\ket{0},\ket{\Xi}}$ and written in this orthonormal basis, our observables take the form
\begin{align*} 
D_R &= -\sigma_z, \\
Q_R &= \cos\theta\,\sigma_z + \sin\theta\,\sigma_x,
\end{align*}
	and using the Khalfin-Tsirelson-Landau identity\footnote{If $A_i^2=B_j^2=\id$, $\chsh^2=4\id-[A_0,A_1]\otimes[B_0,B_1].$} \cite{tsirelson85,landau87b} we can check that 	
\begin{equation}\label{eq:max}
\norm{\chsh(A^+)}= 2\sqrt{1+\sin^2\theta}.
\end{equation}
	We remark that the states that attain this violation can be easily found by diagonalizing the CHSH operator restricted to $\operatorname{span}\DE{\ket{0},\ket{\Xi}}^{\otimes 2}$, that is, a $4\times4$ matrix.  

	With equation \eqref{eq:max}, we see that the maximal CHSH violation $2\sqrt{2}$ \cite{tsirelson80} can be reached iff $\int_{A^+} \varphi_0^2 = 1/2$, and for these binnings, CHSH eigenstates are
\begin{subequations}\label{eq:pipm}
	\begin{equation}
	\ket{\pi_\pm} := \frac{\sqrt{2\pm\sqrt{2}}}{2}\ket{\psi_+} \mp \frac{\sqrt{2\mp\sqrt{2}}}{2} \ket{\phi_-},
	\end{equation}
	with $\exval{\pi_\pm}{\chsh}{\pi_\pm} = \pm2\sqrt{2}$, where
	\begin{equation}
	\ket{\psi_+}:=\frac{\ket{\Xi 0}+\ket{0\Xi}}{\sqrt{2}},\quad\ket{\phi_-}:=\frac{\ket{00}-\ket{\Xi\Xi}}{\sqrt{2}}.
	\end{equation}
	\end{subequations}

	\section{Physically motivated classes of states}\label{sec:restriction}

	In the previous section, we showed that maximal CHSH violation is attainable. In order to understand the state $\ket{\Xi}$, it is useful to explore the binning $A^+ = \R^+$, where the integrals in $\ket{\Xi}$ are easy to solve. We can then write explicitly \[ \ket{\Xi_{\mathbb{R}^+}} = 2\sum_{n=0}^\infty \frac{(-1)^n (2n)!}{\sqrt{2\pi(2n+1)!} \; 2^n n!}\ket{2n+1}. \]
	Its asymptotic is then given by
	\[ \ket{\Xi_{\mathbb{R}^+}} \sim \frac{2^{3/4}}{\sqrt{\zeta\de{\frac{3}{2}}\de{2\sqrt{2}-1}}} \sum_{n=0}^\infty \frac{(-1)^n}{(2n+1)^{\frac{3}{4}}}\ket{2n+1}.\]
	First note that in this state only the odd Fock states appear. This comes from the fact that $\phi_0\phi_n$ is an even function for even $n$, which makes $\int_{\mathbb{R}^+}\phi_0\phi_n= 1/2 \int_{\mathbb{R}}\phi_0\phi_n= 0$ due to the orthogonality of the Hermite functions. Second, we  see that its representation in the Fock basis goes \emph{polynomially} to zero, causing problems like the divergence of the mean number of photons $\exval{\Xi_{\mathbb{R}^+}}{N}{\Xi_{\mathbb{R}^+}}$. In fact, we have numerical evidence that this divergence occurs for \emph{any} choice of binning, forcing us to conclude that states defined in the subspace $\DE{\ket{0},\ket{\Xi}}^{\otimes 2}$ are unfeasible. From now on we look for restricted but physically sound families of states. 

	\subsection{Truncated Fock spaces}

	As a first example we calculate numerically the largest violation of the CHSH inequality given a maximum number $N$ of photons per mode. The results are shown in table \ref{tab0toN}. We chose to omit the states, since our objective is to illustrate the relation between the attainable violation and the size of the subspace.

	\begin{table}[ht]
	\begin{tabular}{ccc}
		Subspace & $\abs{\mean{\chsh}}$  &Set $A^+$\\
		$\Hi_1$ & $2.29$ & $[-0.10,\infty)$\\
		$\Hi_2$ & $2.46$ & $[-0.08,\infty)$ \\
		$\Hi_4$ & $2.56$ & $[-0.05,\infty)$ \\
		$\Hi_{6}$ & $2.61$ & $[-0.04,\infty)$ \\
		$\Hi_{12}$ & $2.67$ & $\R^+$ \\
		$\Hi_{18}$ & $2.70$ & $\R^+$ \\
		$\Hi_{36}$ & $2.74$ & $\R^+$ \\
		$\Hi_{100}$ & $2.77$ & $\R^+$ \\
	\end{tabular}
	\caption{Numercial maximal violation of CHSH for states in the subspace $\Hi_N = \DE{\ket{0},\ldots,\ket{N}}^{\otimes2}$.}
	\label{tab0toN}
	\end{table}

	\subsection{\texorpdfstring{The $N00N$ states}{The N00N states}}

	Another natural restriction is to consider the so-called $N00N$ states, defined as $\ket{N00N}:=(\ket{N0}+\ket{0N})/\sqrt{2}$. In fact, the particular case of $N=2$ was studied before in \cite{cavalcanti10}. The use of $N00N$ states puts some constraints in the expected value of the CHSH operator, particularly because $\mean{D\otimes D}=-1$ for these states. This constraint (perhaps counterintuitively) forces\footnote{We thank N. Brunner for pointing out this fact to us.} $\abs{\mean{\chsh}}\leq 5/2$. A proof of this fact is shown in the appendix \ref{sec:skrzypczyk}.

	The expected value of the CHSH operator for $N00N$ states is given by
	\begin{multline*}
	\mean{\chsh}_{N00N} = 2 + 4\de{\int_{A^+}\varphi_0\varphi_N}^2 \\- 4\int_{A^+} \varphi_N^2\de{1-\int_{A^+}\varphi_0^2}.
	\end{multline*}
	We proceed to show an upper bound to this expected value. First note that by the orthogonality of the Hermite functions
	\begin{equation}\label{maxap}
	\max_{A^+} \int_{A^+}\varphi_0\varphi_N = \frac{1}{2}\int_{\mathbb{R}} \abs{\varphi_0\varphi_N}.
	\end{equation}
	Also, $\int_{A^+} \varphi_N^2\de{1-\int_{A^+}\varphi_0^2}$ is always a nonnegative number, so 
	\begin{equation}\label{maxnoon}
		\mean{\chsh}_{N00N} \leq 2 + \de{\int_{\mathbb{R}} \abs{\varphi_0\varphi_N}}^2.			
	\end{equation}	

	For odd $N$ one can check that for a set $A^+ \subseteq \mathbb{R}$ that maximizes $\int_{A^+}\varphi_0\varphi_N$, $\int_{A^+} \varphi_0^2 = \int_{A^+}\varphi_N^2 = 1/2$ holds, so 
	\[
	\mean{\chsh}_{N00N} \leq 1 + \de{\int_{\mathbb{R}} \abs{\varphi_0\varphi_N}}^2.
	\]
	However $\int_{\mathbb{R}} \abs{\varphi_0\varphi_N} <1$ for all $N>0$, which suggests that it is impossible to violate CHSH with odd $N$, as checked numerically up to $N=7$. 

	For even $N$ we have found violations only for $N = 2$ and $N = 4$. The numerical maximal violations and the upper bound given by equation \eqref{maxnoon} are summarized in table \ref{tabMaxNOON}.
	
	\begin{table}[ht]
		\begin{tabular}{ccc}
			$N$ & $\abs{\mean{\chsh}}$ & Upper bound\\
			$2$ & $2.25$ & $2+\frac{4}{\pi e}$\\
			$4$ & $2.02$ & $2+\frac{4}{\pi e^3}\de{\sqrt{3}+3 \cosh\sqrt{6}-\sqrt{6} \sinh\sqrt{6}}$\\
			$6$ & $2$ & $\approx 2.26$ \\
		\end{tabular}
		\caption{Numerically found violations and analytical upper bounds for $N00N$ states.}\label{tabMaxNOON}
	\end{table}

	Our results shows that in the proposed scenario, $N00N$ states with high number of photons are not useful.

	\subsection{\texorpdfstring{States in the $\DE{\ket{0},\ket{N}}^{\otimes2}$ subspaces}{States in the \{0,N\} subspaces}}

	As a generalization of the $N00N$ states we considered states in the $\DE{\ket{0},\ket{N}}^{\otimes2}$ subspaces. The results are shown in table \ref{tab0N}, and the states in appendix \ref{sec:noonstates}.

	\begin{table}[ht]
	\begin{tabular}{llcc}
		State & Subspace & $\abs{\mean{\chsh}}$  &Set $A^+$\\
		$\ket{\chi_1}$ & $\DE{\ket{0},\ket{1}}^{\otimes2}$ & $2.29$ & $[-0.10,\infty)$\\
		$\ket{\chi_2}$ & $\DE{\ket{0},\ket{2}}^{\otimes2}$ & $2.34$ &$[-0.73,0.73]$\\
		$\ket{\chi_3}$ & $\DE{\ket{0},\ket{3}}^{\otimes2}$ & $2.09$ &$[0.10,1.17]$\\
		$\ket{\chi_4}$ & $\DE{\ket{0},\ket{4}}^{\otimes2}$ & $2.11$ &$[-0.49,0.49]$
	\end{tabular}
	\caption{Maximal violation of CHSH for states in the subspace  $\DE{\ket{0},\ket{N}}^{\otimes2}$.}
	\label{tab0N}
	\end{table}

	We see that the highest violation lies in the subspace $\DE{\ket{0},\ket{2}}^{\otimes2}$, and that the violation seems to decrease with $N$. To understand this result it is enlightening to look at the matrix representation of $Q$ restricted to these subspaces. It is a $2 \times 2$ matrix, with off-diagonal elements
	\[\int_{A^+}\varphi_0\varphi_N.\]
	Using equation \eqref{maxap} and making the asymptotic expansion of the \emph{rhs}
	\[ \int_{\mathbb{R}} \abs{\varphi_0\varphi_N} \sim \frac{2}{\pi}\sqrt[4]{\frac{8}{N\pi}}, \]
	we see that the off-diagonal elements are monotonically decreasing with $N$ and have limit $0$. So for large $N$ the observables $Q$ and $D$ are diagonal in the same basis, so they commute and there is no violation.

	This fact can be understood physically as the increasing distinguishability between $\varphi_0$ and $\varphi_N$ by the $Q$ measurement.

	\subsection{Cat-like states}

	Another idea is to approximate directly the maximally violating states \eqref{eq:pipm}. To do so we repeat their structure but replace the problematic $\ket{\Xi}$ with a well-behaved cat state \cite{jeong05}. The scheme is
	\begin{subequations}\label{eq:catoon}
	\begin{align}
	\ket{\Xi} & \mapsto \ket{\text{cat}} := \frac{\ket{\alpha}\pm\ket{-\alpha}}{\sqrt{2}\sqrt{1\pm e^{-2\abs{\alpha}^2}}},\\
	\ket{\psi_+} &\mapsto \ket{\text{cat}00\text{n}} := \frac{\ket{\text{cat}0}+\ket{0\text{cat}}}{\sqrt{2}},\\
	\ket{\phi_-} &\mapsto \ket{00\text{cat}} := \frac{\ket{00}-\ket{\text{cat}\text{cat}}}{\sqrt{2}},\\
	\ket{\pi} &\mapsto \ket{\Gamma_\pm} := \cos(\theta)\ket{\text{cat}00\text{n}} + \sin(\theta)\ket{00\text{cat}},
	\end{align}
	\end{subequations}
	where $\ket{\Gamma_+}$ is constructed with the even cat, and $\ket{\Gamma_-}$ is constructed with the odd cat. $\ket{\alpha}$ is the coherent state, defined as \[\ket{\alpha} := e^{-\abs{\alpha}^2/2}\sum_{n=0}^\infty \frac{\alpha^n}{\sqrt{n!}}\ket{n}.\] We now have two free parameters to optimize, $\theta$ and $\alpha$. The maximal violation for $\ket{\Gamma_+}$ is $\mean{\chsh} \approx -2.45$, reached with $\theta \approx 1.05$ and $\alpha \approx 2.06i$. For $\ket{\Gamma_-}$, the violation is $\mean{\chsh} \approx -2.51$, reached with $\theta \approx 1.18$ and $\alpha \approx 1.15i$.

\section{Requirements for closing the detection loophole}\label{sec:loophole}

	In the previous section, we studied the maximal attainable violation in the case of pure states and perfect measurements. Now, we shall consider a more realistic scenario, that includes losses and imperfect detections. In particular, we investigate the requirements needed to close the detection loophole. 

	Our approach splits the problem concerning the overall detection efficiency in two parts: the transmittance between the source and the detectors and the efficiency of the detectors. We are also going to consider an asymmetric measurement scenario: photodetectors with efficiency $\eta$ and homodyne measurements with efficiency $1$; after all, the main reason for using homodyne measurements in Bell tests is that they can be made very efficient. Note that this scenario is very similar to the observable-asymmetric scenario proposed in \cite{garbarino10}, where Garbarino found out that if the detection efficiency of one observable is $1$, the efficiency of the other can arbitrarily small and still produce a loophole-free Bell violation.

	In the following, we calculate the critical photodetector efficiency, dark count rates and transmittance required to guarantee a CHSH violation free of the detection loophole. This will be done by restricting our measurement operators to given subspaces and then numerically finding the optimal states.

\subsection{Photodetection efficiency}\label{sec:eta}

	We model the effect of having photodetectors with efficiency $\eta \le 1$ considering that the detection of each photon is an independent event \cite{mandel95}. So the probability that a photodetector clicks $(+)$ for the state $\ket{n}$ is just the complement of the probability that it fails to click for all photons. That is, 
	\[ p(+) = 1 - (1-\eta)^n. \]
In our scheme, this amounts to modifying the measurement operator $D$ by generalising its projectors to POVM elements:
	\begin{subequations}
	\begin{align}
		P_{D+} &\mapsto E_+:=\sum_{n=1}^{\infty}(1-(1-\eta)^n)\ket{n}\bra{n}, \\
		P_{D-} &\mapsto E_-:=\ket{0}\bra{0}+\sum_{n=1}^{\infty} (1-\eta)^n\ket{n}\bra{n}.
	\end{align}
	\end{subequations}
	So now we have $D_\eta:=E_+-E_-$. We remark that this new measurement is not projective anymore, so its outputs are not the eigenvalues of an observable. But the rules for the expected value are the same, so the maximal violation of the CHSH operator is still given by $\norm{\chsh(A^+,\eta)}$.

	In order to analyse the effects of inefficiency, we define a critical efficiency by 
	\[ \eta_c(A^+):=\inf_{\eta}{\{\eta:\norm{\chsh(A^+,\eta)} > 2\}}.\]

	Table \ref{CHSHinef} displays the states that minimize $\eta_c$ for given subspaces, their CHSH values for $\eta=1$ and their $\eta_c$. The states are in the appendix \ref{badstates}. 
	\begin{table}[ht]
	\begin{tabular}{llccc}
		State & Subspace & $\abs{\mean{\chsh}}$ & $\eta_c$ & Set $A^+$\\
           	$\ket{\psi_{2}}$ & $\{\ket{0},\ket{2}\}^{\otimes2}$& $2.037$ & $0.48$ & $[-1.13,1.13]$\\
		$\ket{\psi_{4}}$ & $\{\ket{0},\ket{2},\ket{4}\}^{\otimes2}$ & $2.109$ & $0.36$ & $[-0.90,0.90]$\\
		$\ket{\psi_{6}}$ & $\{\ket{0},\ket{2},\ket{4},\ket{6}\}^{\otimes2}$ & $2.170$ & $0.29$ & $[-0.77,0.77]$\\
		$\ket{\psi_{8}}$ & $\{\ket{0},\ket{2},\ket{4},\ket{6},\ket{8}\}^{\otimes2}$ & $2.212$ & $0.25$ & $[-0.70,0.70]$\\
	\end{tabular}
	\caption{Violation of the CHSH inequality for the states that attain the minimum detection efficiency $\eta_c$ for a given subspace. The values of $\eta_c$ are also given.}
	\label{CHSHinef}
	\end{table}
	
	We presented only subspaces that have even number of photons, because including odd Fock states does not lower the critical $\eta$, even though they do increase $\abs{\mean{\chsh}}$. Note that a large violation does not imply a small $\eta_c$. As an example, the state that maximizes the violation in subspace $\{\ket{0},\ket{2}\}^{\otimes2}$ has $\eta_c=0.66$.

	Finally, we calculated the critical efficiency of $\ket{\pi_-}$, equation \eqref{eq:pipm}, for the binnings $A^+=R^+$ and $A^+=[-\erf^{-1}1/2,\erf^{-1}1/2]$, which are $0.26$ and $0.55$, respectively, and also of the state $\ket{\Gamma_+}$, equation \eqref{eq:catoon}, which is $0.32$ for $\theta=1.12$ and $\alpha~=~2.36i$.

\subsection{Dark counts}

	It is important to notice that an efficiency $\eta<1$ does not affect measurements of the vacuum state. As a consequence, states with optimal $\eta_c$ for a given subspace have a very large amplitude in the $\ket{00}$ component, which implies large sensitivity to dark counts. To model dark counts, we assume that given the state $\ket{0}$, the photodetector has a probability $\delta$ to give the correct outcome $-1$, and probability $1-\delta$ to give the spurious outcome $+1$.

	Using the same ideas of the last section, we construct new POVM elements to model dark counts:
	\begin{subequations}
	\begin{align}
		F_+&:=\sum_{n=1}^{\infty}(1-(1-\eta)^n)\ketbra{n}{n} + (1-\delta) \ketbra{0}{0}, \\
		F_-&:=\delta\ketbra{0}{0}+\sum_{n=1}^{\infty} (1-\eta)^n\ketbra{n}{n}.
	\end{align}
	\end{subequations}
	So now we can generalize $D_\eta$ to $D_{\eta,\delta}:=F_+-F_-$, and as before define $\chsh(A^+,\eta,\delta)$.
		
	With it, we calculate the minimum $\delta$ for the states in table \ref{CHSHinef}, and find out that they are extremely sensitive to dark counts, since they have a very large vacuum amplitude. We thus look for states that are more robust to dark counts and still allow reasonable values for $\eta$ and CHSH violation (see table \ref{tab:goodstates}).   These states have smaller vacuum amplitude, higher entanglement, and higher CHSH violation than the ones presented in the previous section. 

	\begin{table}[ht]
	\begin{tabular}{ccccc}
		State &  Set $A^+$ & $\abs{\mean{\chsh}}$ & $\eta_{\phi}$ & $\delta_{\phi}$ \\
           	$\ket{\phi_{2}}$ & $[-0.66,0.66]$ & $2.30$ & $0.65$ & $0.92$ \\
		$\ket{\phi_{4}}$ & $[-0.49,0.49]$ &  $2.23$ & $0.45$ & $0.94$ \\
		$\ket{\phi_{6}}$ & $[-0.41,0.41]$ &  $2.20$ & $0.34$& $0.95$ \\
		$\ket{\phi_{8}}$ & $[-0.31,0.31]$ &  $2.15$ & $0.28$ & $0.96$\\
	\end{tabular}
	\caption{Examples of states achieving reasonable values for $\delta$, $\eta$, and $\abs{\mean{\chsh}}$. The explicit form of these states are shown in appendix \ref{sec:goodstates}.}
	\label{tab:goodstates}
	\end{table}

	In figure \ref{etac} we show the behaviour of $\abs{\mean{\chsh}}$ as a function of $\eta$ for these states.

		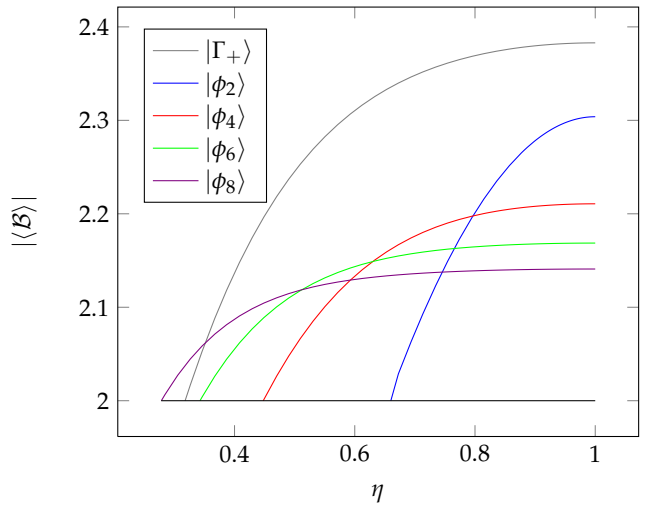
\begin{figure}[ht]
		\centering
	\begin{tikzpicture}
	\begin{axis}[
	xlabel=$\eta$,
	ylabel=$\abs{\mean{\chsh}}$,
	legend style={anchor=north west,cells={anchor=west},at={(0.05,.95)}}
	]

	\addplot[color=gray,mark=none] plot file {cat00neta};	\addlegendentry{$\ket{\Gamma_+}$}
	\addplot[color=blue,mark=none] plot file {phi2eta}; 	\addlegendentry{$\ket{\phi_{2}}$}
	\addplot[color=red,mark=none] plot file {phi4eta}; 	\addlegendentry{$\ket{\phi_{4}}$}
	\addplot[color=green,mark=none] plot file {phi6eta};	\addlegendentry{$\ket{\phi_{6}}$}
	\addplot[color=violet,mark=none] plot file {phi8eta};	\addlegendentry{$\ket{\phi_{8}}$}
	\addplot[black] coordinates {(0.278,2) (1,2)};

	\end{axis}
	\end{tikzpicture}
	\caption{$\abs{\mean{\chsh}}$ as a function of efficiency $\eta$ for $\ket{\Gamma_+}$ ($\theta=1.12$, $\alpha~=~2.36i$) and the states presented in table \ref{tab:goodstates}. The parameters in $\ket{\Gamma_+}$ were optimized to minimize the photodetection efficiency $\eta$ required for a CHSH violation.}
	\label{etac}
	\end{figure}

	\subsection{Transmittance}\label{sec:trans}

	Finally we study the effect of having a channel with transmittance $t\leq1$ connecting the source of the photons to the detectors. We model this effect as an amplitude damping channel \cite{chuang00}  
	\[\mathcal{E}(\rho)=\sum_k F_k\rho F_k^*,\]
	where
	\[F_k=\sum_{n=k}^N \sqrt{\binom{n}{k}} \sqrt{t^{n-k}(1-t)^k} \ketbra{n-k}{n}.\] 
The duality relation
		\[\tr\de{\chsh \de{\sum_{k,l} E_{kl}\rho E_{kl}^*}} = \tr\de{\de{ \sum_{k,l} E_{kl}^*\;\chsh \;E_{kl}}\rho},\]
	where $E_{kl}=F_k\otimes F_l$, allows us to define
	\[\chsh(A^+,\eta,t):=\sum_{kl} E_{kl}^* \; \chsh \; E_{kl}.\] 
	Now we can, as before, define the critical transmittance
	\[
	t_c(A^+):=\inf_{t}{\{t:\norm{\chsh(A^+,\eta=1,t)} > 2\}.}
	\]
	We found numerically $t_c$ for given subspaces (see table \ref{tab:trans}) and present the respective states in appedix \ref{sec:transstates}. Moreover, in table \ref{tab:transgoodstates} we show the minimum $t$ for $\ket{\Gamma}$ and the states presented in table \ref{tab:goodstates} (see figure \ref{fig:t}).
\begin{table}[ht]
	\begin{tabular}{llcccc}
		State & Subspace & $\abs{\mean{\chsh}}$ & $\eta_\xi$ & $t_c$ & Set $A^+$\\
	        $\ket{\xi_{2}}$ & $\{\ket{0},\ket{2}\}^{\otimes2}$& $2.18$ & $0.57$ & $0.78$ &$[-0.95, 0.95]$\\
	        $\ket{\xi_{4}}$ & $\{\ket{0},\ket{2},\ket{4}\}^{\otimes2}$ & $2.18$ & $0.57$ & $0.75$ &$[-0.95,0.95]$\\
	        $\ket{\xi_{6}}$ & $\{\ket{0},\ket{2},\ket{4},\ket{6}\}^{\otimes2}$ & $2.13$ & $0.58$& $0.74$ &$[-0.95,0.95]$\\
	        $\ket{\xi_{8}}$ & $\{\ket{0},\ket{2},\ket{4},\ket{6},\ket{8}\}^{\otimes2}$ & $2.07$ & $0.59$ & $0.74$& $[-0.95,0.95]$\\
	\end{tabular}
	\caption{Violation of the CHSH inequality for the states that attain the minimum transmittance $t_c$ for a given subspace and their minimum $\eta$. As in the inefficiency analysis, the inclusion of  odd numbers of photons does not lower the critical transmittance. These states are presented in appendix \ref{sec:transstates}.}\label{tab:trans}
	\end{table}

\begin{table}[ht]
	\begin{tabular}{ccccc}
		State & Set $A^+$ & $\abs{\mean{\chsh}}$ & $\eta_{\phi}$ & $t_{\phi}$ \\
		$\ket{\Gamma_+}$ & $[-0.48,0.48]$ & $2.38$ & $0.38$ & $0.88$ \\
           	$\ket{\phi_{2}}$  & $[-0.66,0.66]$ & $2.30$ & $0.65$ & $0.81$ \\
		$\ket{\phi_{4}}$  & $[-0.49,0.49]$ & $2.23$ & $0.45$ & $0.87$ \\
		$\ket{\phi_{6}}$  & $[-0.41,0.41]$ & $2.20$ & $0.34$& $0.91$ \\
		$\ket{\phi_{8}}$  & $[-0.31,0.31]$ & $2.15$ & $0.28$ & $0.95$\\
	\end{tabular}
	\caption{Transmittance of $\ket{\Gamma_+}$ and the states presented in table \ref{tab:goodstates}.}\label{tab:transgoodstates}
	\end{table}

\begin{figure}[ht]
		\centering
	\begin{tikzpicture}
	\begin{axis}[
	xlabel=$t$,
	ylabel=$\abs{\mean{\chsh}}$,
	legend style={anchor=north west,cells={anchor=west},at={(0.1,.9)}}
	]

	\addplot[color=gray,mark=none] plot file {cat00nt};	\addlegendentry{$\ket{\Gamma_+}$}
	\addplot[color=blue,mark=none] plot file {phi2t}; 	\addlegendentry{$\ket{\phi_{2}}$}
	\addplot[color=red,mark=none] plot file {phi4t}; 	\addlegendentry{$\ket{\phi_{4}}$}
	\addplot[color=green,mark=none] plot file {phi6t};	\addlegendentry{$\ket{\phi_{6}}$}
	\addplot[color=violet,mark=none] plot file {phi8t};	\addlegendentry{$\ket{\phi_{8}}$}
	\addplot[black] coordinates {(1,2) (0.812,2)};

	\end{axis}
	\end{tikzpicture}
	\caption{$\abs{\mean{\chsh}}$ as a function of transmittance $t$ for $\ket{\Gamma_+}$($\alpha=1.91$, $\theta=1.03$) and the states presented in table \ref{tab:goodstates}. The parameters in $\ket{\Gamma_+}$ were optimized to minimize the transmittance efficiency $t$ required for a CHSH violation.}
	\label{fig:t}
	\end{figure}
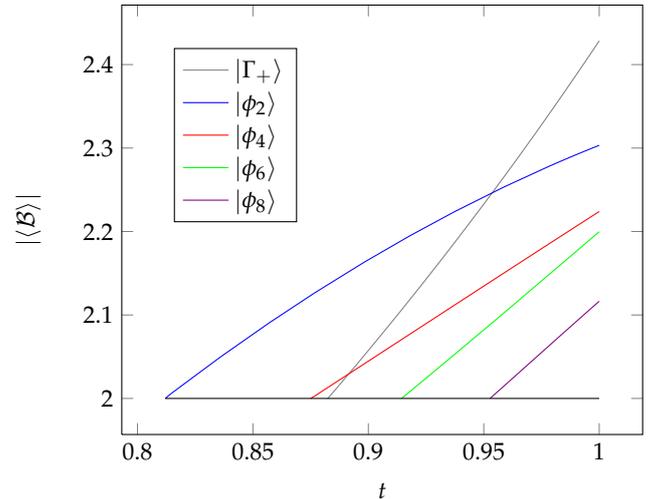

\section{Multipartite states}\label{sec:multipartite}

	Multipartite states can also be seen as interesting candidates for loophole-free Bell tests \cite{ecava07,acin09}. For instance the $N$-mode GHZ state $\ket{GHZ}=(\ket{0}^{\otimes N}+\ket{1}^{\otimes N})/\sqrt{2}$ was shown to attain an exponential violation of the $N$-partite Mermin inequality when only homodyne measurements are used \cite{acin09}. Those measurements are given by two orthogonal quadratures $X$ and $P$ followed by a sign binning process (\ie\ $A^+=B^+=...=N^+=\R^+$). 

	This result can be easily recovered within the framework developed here. As commented in section \ref{sec:methods}, in the present situation the measurement operators are proportional to Pauli $\sigma_x$ and $\sigma_y$ measurements. Noting that these operators are the optimal operators used in the violation of the Mermin inequality with the GHZ state, the violation of $\ket{GHZ}$ is given simply by 
	\[\de{\frac{2}{\pi}}^{N/2} 2^{\frac{N+1}{2}}.\]
	This is nothing but the standard GHZ violation multiplied by the term $(2/\pi)^{N/2}$ which comes from the norm of the measurement operators. 

	We have also considered the three-mode state $\ket{W} = \ket{001}+\ket{010}+\ket{100}$ and the Mermin inequality \cite{mermim80,zukowski02}
	\[
	\abs{\mean{DQQ}+\mean{QDQ}+\mean{QQD}-\mean{DDD}}\leq2.
	\]  
	We have found a violation of this inequality of $1+\frac{4}{\pi} \approx 2.29$ for $A^+=B^+=C^+=\R^+$. Unfortunately the minimum detection efficiency required in this case is $\eta_c=0.86$.

\section{Discussion and future directions}\label{sec:discussion}

	We studied CHSH inequalities that combine homodyne measurements and photodetection, where the quantum information is encoded in two modes of the electromagnetic field. First, we showed the maximum attainable violation for a given binning (set $A^+$). With this relation, we proved that maximal violation is possible in this hybrid scenario and characterized the family of states that attains it. Then we proceeded to seek states that had a good combination of feasibility, high violation, small efficiency and transmittance requirements, and were also robust to dark counts.

	Using a simple numerical technique we found the minimum photodetection efficiency necessary to obtain a violation given a limitation on the maximum number of photons in each channel. We found states that attained violation for photodetection efficiencies as small as $0.28$. In the same direction, we also showed the possibility of violations with transmittances of the order of $0.75$.

	We presented a state  ($\ket{\Gamma}$ in \eqref{eq:catoon}) that had the best combination we found of feasibility, reasonably high violation ($2.38$) and small efficiency and transmittance requirements ($\eta >0.32$ or $t >0.92$). This state is made of vacuum and cat-like superpositions of coherent states and therefore perfectly physical, although by no means easy to produce.

	Finally, we made a brief analysis of the multipartite scenario. Using the Mermin inequality we recovered the result of \cite{acin09} for the violations of the $N$-mode GHZ state and explored the tripartite $\ket{W}$ to find violations for $\eta>0.86$.

	The results presented here greatly enhance the possibilities of attaining Bell violations in this experimental setup. In particular, the reasonably low requirements both in detection efficiency and transmittance are rather promising. At this point, the greatest experimental quest seems to be the search for feasible states and their eventual realization in the lab. In that regard, there is lots of room for improvement, since the cases studied here represent just a small fraction of all the possibilities. 

	From a theoretical point of view, a natural development would be to improve our results by using other Bell inequalities, such as $Innmm$, as done in \cite{Tamas-Stefano-Nicolas}. Another approach would be to follow the work of Garbarino \cite{garbarino10} and find a state which requires vanishing photodetector efficiency to provide a loophole-free Bell violation.

\begin{acknowledgments}
The authors would like to  thank A. Cabello, F. Brandão, R.~O. Vianna, V. Scarani, and N. Brunner for helpful discussions. This work was supported by 
the Brazilian agencies Fapemig, Capes, CNPq, and INCT-IQ, the National Research Foundation and the Ministry of Education of Singapore.
\end{acknowledgments}
\appendix

\section{\texorpdfstring{$\norm{\chsh}=2\sqrt{1+\sin^2{\theta}}$}{||B||=2sqrt{1+sin^2t}}} \label{sec:russo}

	\begin{lemma}
		$D$ and $Q$ can be written as \[D=\Pi D\Pi+(\id-\Pi)D(\id-\Pi)\] 
					      \[Q=\Pi Q\Pi+(\id-\Pi)Q(\id-\Pi),\]
	where $\Pi$ is the projector onto the subspace generated by $\{\ket{0},Q\ket{0}\}$.
	\end{lemma}
	\begin{proof}
	Note that $\{\ket{0},Q\ket{0}\}$ is an invariant subspace of both operators $D$ and $Q$, as
	\[Q\de{\alpha\ket{0}+\beta Q\ket{0}} = \alpha Q\ket{0} + \beta\ket{0}\]
	and 
	\begin{align*}
	D\de{\alpha\ket{0}+\beta Q\ket{0}} &= -\alpha \ket{0} + \beta DQ\ket{0} \\
					   &= -\alpha \ket{0} + \beta \de{-2\ketbra{0}{0}Q\ket{0}+Q\ket{0}} \\
					   &= -(\alpha+2\beta\exval{0}{Q}{0})\ket{0} + \beta Q \ket{0}.
	\end{align*}

	Since both $D$ and $Q$ are self-adjoint, it follows that the pre-image of $\{\ket{0},Q\ket{0}\}$ is also within $\{\ket{0},Q\ket{0}\}$, so the orthogonal decomposition is valid for both operators.
	\end{proof}

	Using this lemma one can check that\footnote{Just use the fact that $(\id-\Pi)D(\id-\Pi)=\id$.} $[Q,D]=[\Pi Q\Pi,\Pi D\Pi]$. That is, the only subspace relevant for a CHSH violation is the one generated by $\{\ket{0},Q\ket{0}\}$. Now we can restrict the domain of our operators to it and calculate the maximal attainable violation using Tsirelson's identity~\cite{tsirelson85}, $\norm{\chsh}^2=4+\norm{\comm{Q}{D}}^2$. Let's now understand how our observables act on this subspace.

        Expanding $Q$ in the Fock basis \eqref{eq:qdef}, we have 
	\[Q\ket{0} = \de{2\int_{A^+}\varphi_0^2-1}\ket{0} + \sum_{n=1}^\infty 2\int_{A^+}\varphi_0\varphi_n\ket{n}.\]
	Since $Q$ is unitary, it is useful to define 
	\begin{align*}
	\cos\theta &:= 2\int_{A^+}\varphi_0^2-1, \; \\
        \ket{\Xi} &:= \frac{1}{\sin\theta}\sum_{n=1}^\infty 2\int_{A^+}\varphi_0\varphi_n\ket{n}
	\end{align*}
	for $\theta \in (0,\pi)$, so that
	\[Q\ket{0} = \cos\theta\ket{0}+\sin\theta\ket{\Xi}.\]
	This allows us to write the restriction of $Q$ in the orthonormal basis $\DE{\ket{0},\ket{\Xi}}$ as:
	\[Q_R = \begin{pmatrix}
			\cos\theta & \sin\theta \\
			\sin\theta & -\cos\theta \\
	\end{pmatrix} = \cos\theta\,\sigma_z + \sin\theta\,\sigma_x.\]
	For the photodetection observable $D$, we simply notice that $D\ket{0}=-\ket{0}$ and $D\ket{\Xi}=\ket{\Xi}$ to see that restricted to the $\{\ket{0},\ket{\Xi}\}$ basis the operator $D$ is 
\begin{equation}
D_R=-\sigma_z.
\end{equation}

	With these forms of $Q_R$ and $D_R$, a straightforward calculation shows that 
	\[\norm{\chsh}^2=4+4\sin^2{\theta}.\]

	\section{\texorpdfstring{Effects of specifying $\mean{A_i\otimes B_j}$}{Effects of specifying <AB>}}
	\label{sec:skrzypczyk}

	\begin{theorem}\label{theo:skrypczyk}
		If $\abs{\mean{A_i \otimes B_j}} = 1$, for any given $i,j$, then $\max \abs{\mean{\chsh}} = 5/2$.
	\end{theorem}
	\begin{proof}
	The proof of this theorem is based on the ideas presented in \cite{dieks02}.

	Define 
	\begin{align*} 
	\ket{A_0} &:= A_0 \otimes \id \ket{\psi}, & \ket{B_0} &:= \id \otimes B_0\ket{\psi}, \\ 
	\ket{A_1} &:= A_1 \otimes \id \ket{\psi}, &  \ket{B_1} &:= \id \otimes B_1 \ket{\psi}.
	\end{align*} 
	So $\norm{\ket{A_i}} = \norm{\ket{B_i}} = 1$ and \[\exval{\psi}{\chsh}{\psi} = \braket{A_0}{B_0} + \braket{A_0}{B_1} + \braket{A_1}{B_0} - \braket{A_1}{B_1}.\]

	Now we choose $\braket{A_0}{B_0}=1$, the proof being the same for other $i,j$. So $\ket{A_0} = \ket{B_0}$ and we can write the expected value of the CHSH operator as 
	\begin{align*}
	\abs{\exval{\psi}{\chsh}{\psi}} &= \abs{ 1 + \braket{B_0}{B_1} + \braket{A_1}{B_0} - \braket{A_1}{B_1}}\\
	&\le \abs{ 1 + \braket{B_0}{B_1}} + \abs{\bra{A_1}\de{\ket{B_0}-\ket{B_1}}}\\
	&\le \abs{ 1 + \braket{B_0}{B_1}} + \norm{\ket{B_0}-\ket{B_1}}\\
	&= \abs{ 1 + \braket{B_0}{B_1}} + \sqrt{2}\sqrt{1-\braket{B_0}{B_1}}\\
	& \le 5/2.
	\end{align*}
	Note that $\braket{B_0}{B_1} = \braket{A_0}{B_1}$ is real, as an expected value of a self-adjoint operator, so we can pass from the third line to fourth.
	\end{proof}

	We can generalize this theorem by fixing the value of $\abs{\mean{A_i \otimes B_j}}$ and optimising with respect to the other correlation terms. By using this framework we can recover the above theorem, prove that if $\abs{\mean{A_i \otimes B_j}}=0$, then $\max \abs{\mean{\chsh}} = 3\sqrt{3}/2\approx 2.60$, or prove that $\abs{\mean{A_i \otimes B_j}} = 1/\sqrt{2}$ for all $i,j$ is a necessary condition for attaining the Tsirelson bound. The general result is presented in figure \ref{mmbound}.

	\begin{figure}[htb]
		\centering
	\begin{tikzpicture}
	\begin{axis}[
	axis x line=center,
	axis y line=center,
	ymin=2.5,ymax=2.849,
	xmin=-1.05,xmax=1.05,
	extra y ticks={2.82843},
	extra y tick labels={$2\sqrt{2}$},
	extra x ticks={-0.707107,0.707107},
	extra x tick labels={$-\frac{1}{\sqrt{2}}$,$\frac{1}{\sqrt{2}}$},
	]

	\addplot[color=blue,mark=none] plot file {mmdmm};

	\end{axis}
	\end{tikzpicture}
	\caption{$\max \abs{\mean{\chsh}}$ as a function of any expected value.}
	\label{mmbound}
	\end{figure}

\section{\texorpdfstring{States referenced in table \ref{tab0N}}{States referenced in table II}}\label{sec:noonstates}

\begin{multline*}
\ket{\chi_{1}}=0.22\ket{00}-0.66(\ket{01}+\ket{10})-0.28\ket{11}
\end{multline*}
\begin{multline*}
\ket{\chi_{2}}=-0.13\ket{00}-0.69(\ket{02}+\ket{20})+0.07\ket{22}
\end{multline*}
\begin{multline*}
\ket{\chi_{3}}=0.28\ket{00}-0.67(\ket{03}+\ket{30})+0.03\ket{33}
\end{multline*}
\begin{multline*}
\ket{\chi_{4}}=0.19\ket{00}-0.69(\ket{04}+\ket{40})-0.02\ket{44}
\end{multline*}

\section{\texorpdfstring{States referenced in table \ref{CHSHinef}}{States referenced in table IV}}\label{badstates}

\begin{multline*}
\ket{\psi_{2}}=0.98\ket{00}+0.17\ket{22}+0.03(\ket{02}+\ket{20})
\end{multline*}
\begin{multline*}
\ket{\psi_{4}}=0.96\ket{00}+0.19\ket{22}+0.05\ket{44}+0.07(\ket{02}+\ket{20})\\
-0.04(\ket{04}+\ket{40})-0.10(\ket{24}+\ket{42})
\end{multline*}
\begin{multline*}
\ket{\psi_{6}}=0.94\ket{00}-0.19\ket{22}-0.08\ket{44}-0.01\ket{66}\\
-0.10(\ket{02}+\ket{20})+0.06(\ket{04}+\ket{40})\\
-0.04(\ket{06}+\ket{60})+0.12(\ket{24}+\ket{42})\\
-0.07(\ket{26}+\ket{62})+ 0.04(\ket{46}+\ket{64})
\end{multline*}
\begin{multline*}
\ket{\psi_{8}} = 0.92\ket{00}-0.17\ket{22}-0.09\ket{44}-0.04\ket{66}+0.01\ket{88}\\
-0.12(\ket{02}+\ket{20})+0.08(\ket{04}+\ket{40})\\
-0.06(\ket{06}+\ket{60})+0.04(\ket{08}+\ket{80})\\
+0.12(\ket{24}+\ket{42})-0.08(\ket{26}+\ket{62})\\
+0.05(\ket{28}+\ket{82})+0.06(\ket{46}+\ket{64})\\
-0.03(\ket{48}+\ket{84})+ 0.01(\ket{68}+\ket{86})
\end{multline*}

\section{\texorpdfstring{States referenced in table \ref{tab:goodstates}}{States referenced in table V}}\label{sec:goodstates}

\begin{multline*}
\ket{\phi_{2}}=0.22\ket{00}-0.69(\ket{02}+\ket{20})-0.01\ket{22}
\end{multline*}
\begin{multline*}
\ket{\phi_{4}}=0.31\ket{00}+0.26(\ket{02}+\ket{20})-0.62(\ket{04}+\ket{40})\\
\end{multline*}
\begin{multline*}
\ket{\phi_{6}}=0.38\ket{00}+0.17(\ket{02}+\ket{20})-0.30(\ket{04}+\ket{40})\\
+0.56(\ket{06}+\ket{60})
\end{multline*}
\begin{multline*}
\ket{\phi_{8}}=0.42\ket{00}+0.13(\ket{02}+\ket{20})-0.20(\ket{04}+\ket{40})\\
+0.32(\ket{06}+\ket{60})-0.51(\ket{08}+\ket{80})
\end{multline*}

\section{\texorpdfstring{States referenced in table \ref{tab:trans}}{States referenced in table VII}}\label{sec:transstates}

\begin{multline*}
\ket{\xi_{2}}=-0.91\ket{00}+0.07(\ket{02}+\ket{20})+0.40\ket{22}
\end{multline*}
\begin{multline*}
\ket{\xi_{4}}=-0.83\ket{00}+0.36\ket{22}-0.40\ket{44}+0.07(\ket{02}+\ket{20})\\
-0.03(\ket{04}+\ket{40})-0.10(\ket{24}+\ket{42})
\end{multline*}
\begin{multline*}
\ket{\xi_{6}}=0.67\ket{00}-0.29\ket{22}+0.30\ket{44}-0.45\ket{66}\\
-0.05(\ket{02}+\ket{20})+0.03(\ket{04}+\ket{40})\\
-0.01(\ket{06}+\ket{60})+0.11(\ket{24}+\ket{42})\\
-0.04(\ket{26}+\ket{62})- 0.25(\ket{46}+\ket{64})
\end{multline*}
\begin{multline*}
\ket{\xi_{8}} = 0.50\ket{00}-0.22\ket{22}+0.22\ket{44}-0.10\ket{66}+0.43\ket{88}\\
-0.03(\ket{02}+\ket{20})+0.02(\ket{04}+\ket{40})\\
-0.01(\ket{06}+\ket{60})+0.01(\ket{08}+\ket{80})\\
+0.09(\ket{24}+\ket{42})-0.04\ket{26}+\ket{62})\\
+0.02(\ket{28}+\ket{82})-0.21(\ket{46}+\ket{64})\\
+0.10(\ket{48}+\ket{84})+ 0.41(\ket{68}+\ket{86})
\end{multline*}

\bibliographystyle{linksen}
\bibliography{biblio}
\end{document}